\documentclass[pdflatex,sn-mathphys-num]{sn-jnl}

\usepackage{graphicx} 
\usepackage{geometry}
\usepackage{amsmath,amssymb,amsfonts}
\usepackage{bm}
\usepackage{bbm}
\usepackage{xcolor}
\usepackage{hyperref}

\newtheorem{proposition}{Proposition}
\newtheorem{lemma}{Lemma}
 
\newtheorem{corollary}{Corollary}
\newtheorem{remark}{Remark} 
\newtheorem{conjecture}{Conjecture} 

\newcommand{\tr}{\mathrm{tr}}
\newcommand{\R}{\mathbb{R}}
\newcommand{\e}{\mathrm{e}}
\newcommand{\dd}{\mathrm{d}}

\begin{document}
\title{Difference equations of average entropies}
\author{Youyi Huang~\!\footnotemark[1],~~Linfeng Wei~\!\footnotemark[2],~~Lu Wei~\!\footnotemark[2],~~and Peter J.~Forrester~\!\footnotemark[3]
\footnotetext[1]{University of Central Missouri, Missouri 64093, USA}
\footnotetext[2]{Department of Computer Science, Texas Tech University, Texas 79409, USA}
\footnotetext[3]{School of Mathematics and Statistics, University of Melbourne, Victoria 3010, Australia}}

\abstract{Exact cumulants of entanglement entropies of random state ensembles have traditionally been studied within the random matrix framework. In this work, we propose an alternative approach based on the intrinsic connection to integrable systems. The central idea is to embed entropic quantities into tau functions satisfying Toda-type lattice equations, which in turn yield linear difference equations for their averages. Directly solving the difference equations recovers exact entropy formulas in the literature. The integrable systems approach bypasses the case-by-case, ensemble-dependent derivations required by random matrix methods. The approach also suggests a possible route towards unified and more efficient higher-order cumulant calculations by exploring integrable hierarchies.}


\maketitle


\section{Introduction}
Mathematical descriptions of entanglement are naturally formulated in terms of models of generic states, which are ensembles of random states constructed from different information-geometric metrics~\cite{BZ17}. Exact entropic statistics for quantifying entanglement, including von Neumann entropy and quantum purity, over various ensembles have been studied extensively since Page's conjectured formula~\cite{Page93} of average von Neumann entropy in 1993. These include the baseline model of Hilbert--Schmidt ensemble~\cite{Zyczkowski2001,Page93,Foong94,Ruiz95,Se96,Giraud07,VPO16,Wei17,Wei20,HWC21,HW26,Wei26RE,FN2026} and its refined version of Bures--Hall ensemble~\cite{Hall1998,Sommers2003,OSZ10,FK16,Sarkar2019,Wei20BHA,Wei20BH,LW21,WHW25,WHW26,Wei26RE}, as well as the recent fermionic Gaussian ensemble ~\cite{BHK21,BHKRV22,HW22,HW23,HW23C} and Bogoliubov--Kubo--Mori ensemble~\cite{Mi25,SW26}. Page's conjecture~\cite{Page93} and its subsequent proofs~\cite{Foong94,Ruiz95,Se96} appear to have firmly embedded the computation of entropic cumulants within the random matrix framework, which, as the above cited references reveal, has in fact dominated the subsequent development.

Despite the progress, the random matrix framework faces fundamental limitations. The framework necessitates case-by-case cumulant computations: the Hilbert--Schmidt ensemble utilizes integrals involving Laguerre polynomials~\cite{HW26}, the Bures--Hall ensemble relies on the connection to a two-matrix model~\cite{FK16}, and the Bogoliubov--Kubo--Mori ensemble exploits factorization properties of its moment matrix~\cite{SW26}. Moreover, each successive cumulant requires substantially greater effort, despite the observed approximately linear growth in the complexity of higher-order cumulant formulas~\cite{HW26,WHW25}. The present work considers a methodological transition to an integrable systems framework, in which the case-by-case calculations can be bypassed. An integrable systems framework incorporates entropic statistics into appropriate tau functions that satisfy Toda-type lattice equations. Such structural relations are otherwise not accessible within the random matrix framework.

As a first step to demonstrate the advantage of the integrable systems framework, we rederive average entropy formulas over different random state ensembles. Specifically, differentiation in the parameters of the constructed tau functions transforms the nonlinear lattice equations into linear difference equations in the matrix dimension $m$ satisfied by the average entropies. These difference equations are solved by elementary methods to recover mean entropy formulas. The rederivations are therefore more than alternative proofs but highlight the potential of integrable systems methods for unified and more efficient higher-order cumulant calculations than random matrix methods.

We begin by introducing the settings from which random state ensembles arise. For an $m\times m$ reduced density matrix $\rho$ of a bipartite system having the spectrum \(\lambda_1,\ldots,\lambda_m\) with the constraint
\begin{equation}\label{eq:rho}
\tr(\rho)=1,
\end{equation}
the degree of entanglement can be measured by different entropies~\cite{BZ17,Majumdar}. These include the von Neumann entropy
\begin{equation}\label{eq:von-neumann-entropy}
S=-\sum_{i=1}^{m}\lambda_i\ln\lambda_i
\end{equation}
supported in $0\leq S\leq\ln m$ and the quantum purity
\begin{equation}\label{eq:purity}
P=\sum_{i=1}^{m}\lambda_i^2
\end{equation} 
supported in $1/m\leq P\leq1$. The Hilbert--Schmidt ensemble~\cite{Page93,Majumdar} serves as a baseline model. The corresponding joint density of eigenvalues $\boldsymbol{\lambda}=\{\lambda_1,\ldots,\lambda_m\}$ of the reduced density matrix \(\rho\) is
\begin{equation}\label{eq:HS}
f(\boldsymbol{\lambda})\propto\delta\left(1-\sum_{i=1}^{m}\lambda_i\right)
\prod_{1\leq i<j\leq m}(\lambda_i-\lambda_j)^2\prod_{i=1}^{m}\lambda_i^{\alpha} \mathbbm{1}_{\lambda_i > 0},
\end{equation}
where the difference between the environment (ancilla) dimension $n$ and the system dimension $m$ is denoted by
\begin{equation}\label{eq:alpha}
\alpha=n-m.
\end{equation}
Note that all the densities involving $\alpha$ throughout this work are defined for arbitrary non-negative real $\alpha$. The parameter $\alpha$ is replaced by the dimension difference~(\ref{eq:alpha}) only when presenting the final average entropy formulas. With the fixed trace condition replaced by $\e^{-\sum_{i=1}^m x_i}$, this gives the eigenvalue density of complex Wishart matrices~\cite{Forrester}
\begin{equation}\label{eq:LUE}
\rho(\mathbf{x})=\frac{1}{Z}\prod_{1\leq i<j\leq m}(x_i-x_j)^2\prod_{i=1}^{m}x_i^{\alpha}\e^{-x_i}\mathbbm{1}_{x_i > 0},
\end{equation}
where 
\begin{equation}\label{eq:LUEZ}
Z=\prod_{i=1}^m\Gamma(i+1)\Gamma(\alpha+i).
\end{equation}
The density~(\ref{eq:LUE}) is also known as the unconstrained Hilbert--Schmidt ensemble, or the Laguerre unitary ensemble~\cite{Forrester}. Hereafter, we adopt the convention of referring to unconstrained (classical) ensembles as constrained (quantum) ensembles, and vice versa, when the meaning is clear from the context.

The Bures--Hall ensemble is an improved version of the Hilbert--Schmidt ensemble obtained by replacing its random state by a certain superposition~\cite{Hall1998,Zyczkowski2001,Sommers2003,Sarkar2019}. The joint eigenvalue density of the Bures--Hall ensemble is given by
\begin{equation}\label{eq:BH}
f(\boldsymbol{\lambda})
\propto
\delta\left(1-\sum_{i=1}^{m}\lambda_i\right)
\prod_{1\leq i<j\leq m}
\frac{(\lambda_i-\lambda_j)^2}
{\lambda_i+\lambda_j}
\prod_{i=1}^{m}\lambda_i^{\alpha-\frac{1}{2}} 
\mathbbm{1}_{\lambda_i > 0}.
\end{equation}
The corresponding unconstrained version is~\cite{Sarkar2019,Wei20BHA,Wei20BH,LW21,WHW25}
\begin{equation}\label{eq:UBH}
\rho({\mathbf {x}})=\frac{1}{Z}\prod_{1\le i<j\le m}
\frac{(x_i-x_j)^2}{x_i+x_j}\prod_{i=1}^m
x_i^{\alpha-\frac{1}{2}}\e^{-x_i}\mathbbm{1}_{x_i > 0}
\end{equation}
with 
\begin{equation}\label{eq:cBH}
Z=2^{-m(2\alpha-1+m)}\pi^{m/2}\prod_{i=1}^{m}\frac{\Gamma(i+1)\Gamma(2\alpha+i)}{\Gamma(\alpha+i)}.
\end{equation}
Since there are exact relations, known as moment conversions~\cite{HW26,FN2026}, of entropies~(\ref{eq:von-neumann-entropy}),~(\ref{eq:purity}) between constrained ensembles $f(\boldsymbol{\lambda})$ and unconstrained ensembles $\rho(\mathbf{x})$, we will focus on the latter ones in what follows.

The essential idea of the integrable systems approach, for which we take guidance from  the study of (multi-) moments of the Gaussian unitary ensemble ~\cite{MS09}, is to take for the unconstrained ensembles $\rho(\mathbf{x})$, the Toda tau functions of the form
\begin{equation}\label{eq:tau-introduction}
\tau_m({\bf t},q; s)=
\frac{Z}{m!}
\int_{\mathbb{R}_+^m}\rho(\mathbf{x})
\prod_{i=1}^m
\e^{q x_i^s+\sum_{r\geq 1}t_rx_i^r}\dd x_i.
\end{equation}
Here, in addition to the Toda times ${\bf t}=(t_1,t_2,\dots)$, we consider the deformation
\begin{equation}
\prod_{i=1}^m\e^{q x_i^s}
\end{equation}
of a real-valued power parameter $s$. This one-body deformation preserves integrable structure of tau functions~(\ref{eq:tau-introduction}) so that they satisfy different types of Toda lattice hierarchies for different densities $\rho(\mathbf{x})$. The $q$-flow is responsible for generating the desired entropic quantities through difference equations of averages of spectral moments 
\begin{equation}\label{eq:Rs}
R_s=\sum_{i=1}^{m} x_i^s,
\end{equation}
of a real-valued $s$ via differentiation of relevant Toda lattice systems with respect to $q$. Differentiation of~(\ref{eq:Rs}) with respect to $s$ before setting $s=1$ defines the induced von Neumann entropy~\cite{Wei17},
\begin{equation}\label{eq:T}
T=\sum_{i=1}^{m}x_{i}\ln x_{i},
\end{equation}
whereas the specialization $s=2$ in~(\ref{eq:Rs}) defines the induced purity. Once the average induced entropies are computed, moment conversions~\cite{HW26,FN2026} yield the desired average entropy formulas.

Besides the major models of Hilbert--Schmidt and Bures--Hall ensembles, we also adapt the integrable systems approach to the Muttalib--Borodin ensemble~\cite{FL15,FW17,Ch18},
\begin{equation}\label{eq:LMB}
\rho({\bf x})=\frac{1}{Z}\prod_{1\leq i<j\leq m}
(x_i-x_j)\left(x_i^\theta-x_j^\theta\right)
\prod_{i=1}^{m}x_i^\alpha\e^{-x_i}\mathbbm{1}_{x_i > 0}
\end{equation}
with $\theta>0$ and
\begin{equation}\label{eq:zMB}
Z=\theta^{m(m-1)/2}\prod_{i=1}^{m}\Gamma(i+1)\,\Gamma\!\left(\alpha+1+\theta(i-1)\right).
\end{equation}
For \(\theta=1\), the ensemble~(\ref{eq:LMB}) reduces to the Hilbert--Schmidt ensemble~(\ref{eq:LUE}). In the limit \(\theta\to0\), it recovers the Bogoliubov--Kubo--Mori ensemble~\cite{Mi25,SW26}
\begin{equation}\label{eq:uBKM}
\rho({\mathbf {x}})=\frac{1}{\Gamma^m(\alpha+1)\prod_{i=1}^{m}\Gamma(i+1)}\prod_{1\le i<j\le m}
(x_i-x_j)\!\left(\ln x_i-\ln x_j\right)\prod_{i=1}^m
x_i^\alpha\e^{-x_i}\mathbbm{1}_{x_i > 0}.
\end{equation}
Therefore, the Muttalib--Borodin ensemble interpolates between the Hilbert--Schmidt ensemble of $\theta=1$ and the Bogoliubov--Kubo--Mori ensemble of $\theta=0$. 

The constrained Bogoliubov--Kubo--Mori ensemble, obtained by restoring $\delta\left(1-\sum_{i=1}^{m}\lambda_i\right)$ while removing $\prod_{i=1}^m\e^{-x_i}$, cf.~(\ref{eq:HS}) and~(\ref{eq:BH}), was proposed in~\cite{Mi25} as a new generic state ensemble. An exact formula of average entropy over such an ensemble was reported recently in \cite{SW26}. By incorporating an additional deformation 
\begin{equation}
\prod_{i=1}^m\e^{u x_i^\theta}
\end{equation}
into the Toda lattice equation~(\ref{eq:tau-introduction}), one can similarly arrive at a difference equation of average entropy of the ensemble as shown in Section~\ref{sec:results}, from which the exact entropy formula in~\cite{SW26} is also reclaimed. 

Our main results in Section~\ref{sec:results} of difference equations of spectral moments~(\ref{eq:Rs}) are in matrix dimensions, while the existing results along these lines~\cite{HT03,HW26,WHW26} mainly concern difference equations in power parameters. The connection between the two types of difference equations is established in Section~\ref{sec:smr} by deriving mixed difference equations in both matrix dimensions and power parameters. Finally, we discuss perspectives of integrable systems approach in efficiently finding higher-order cumulants in Section~\ref{sec:con}. A new cumulant identity employed throughout this work in simplifying joint cumulants is summarized in Appendix~\ref{app1}. 

\section{New difference equations of spectral moments}\label{sec:results}
In this section, we present the main results of this work on difference equations in matrix dimensions of various random matrix models. Specifically, difference equations of spectral moments and corresponding entropies are derived for the Hilbert--Schmidt ensemble~(\ref{eq:LUE}), the Bures--Hall ensemble~(\ref{eq:UBH}), and the Muttalib--Borodin ensemble~(\ref{eq:LMB}) in Section~\ref{sec:HS}, Section~\ref{sec:BH}, and Section~\ref{sec:MB}, respectively. The new results on difference equations of spectral moments, summarized in the following three propositions, may be of independent interest to other applications.

We first set up some notations and discuss the general procedure. As will be seen, to compute average entropies, it is sufficient to consider a single time
\begin{equation}\label{eq:tauts}
\mathbf{t}=(t,0,0,\ldots)
\end{equation}  
in the tau function~(\ref{eq:tau-introduction}) that becomes
\begin{equation}\label{eq:tau-integral}
\tau_m(t,q;s)=\frac{Z}{m!}\int_{\R_+^m}\rho(\mathbf{x})\prod_{i=1}^{m}\e^{t x_i+q x_i^s}\dd x_i.
\end{equation}

The integrable structure of the tau functions is a key mechanism leading to the desired difference equations. For Hilbert--Schmidt ensemble~(\ref{eq:LUE}) and Bures--Hall ensemble~(\ref{eq:UBH}), the tau function~(\ref{eq:tau-integral}) satisfies the Toda lattice equation~\cite{UenoTakasaki1984,AdlerVanMoerbeke1995} and the Pfaff Toda lattice equation~\cite{HuLi17,Chang18}, respectively. By taking derivatives of these nonlinear lattice equations and introducing
\begin{equation}\label{eq:spt}
M_m(s;t)=\left.\frac{\partial}{\partial q}\ln\tau_m(t,q;s)\right\rvert_{q=0},
\end{equation} 
one obtains difference equations of average spectral moment~(\ref{eq:Rs}) by setting \(t=0\) in~(\ref{eq:spt}) as
\begin{equation}\label{eq:spM} 
M_m(s)=M_m(s;0)=\mathbb{E}\!\left[R_s\right].
\end{equation}
We emphasize that the spectral moment parameter $s$ in~(\ref{eq:spM}) takes real values, and hence all results obtained in this work hold for real $s$. This is distinct from existing results in the literature~\cite{HT03, GGR20, LV11}, which are restricted to integer values of $s$.

Difference equations of average purity are now readily available by setting \(s=2\) in the $M_m(s)$ recurrences, while difference equations of average von Neumann entropy ${M}_m$ are obtained via 
\begin{equation}\label{eq:unconstrained-entropy}
{M}_m=\left.\frac{\partial}{\partial s}M_m(s)\right\rvert_{s=1}.
\end{equation}
The resulting difference equations can be straightforwardly solved. The solutions recover known formulas in the literature~\cite{Page93,Foong94,OSZ10,Sarkar2019,Wei20BHA,LW21,FN2026,SW26} of the average von Neumann entropy and the average purity. 

For the Muttalib--Borodin ensemble~(\ref{eq:LMB}), the derivation of difference equations follows a similar procedure. The only new ingredient is an additional $u$-flow of the tau function that accounts for the $\theta$-deformation of the ensemble
\begin{equation}\label{eq:MB-tau-integral}
\tau_m(t,q,u;s,\theta)=\frac{Z}{m!}\int_{\R_+^m}\rho(x)\prod_{i=1}^{m}\e^{tx_i+qx_i^s+u x_i^\theta}\dd x_i,
\end{equation}
which will be shown to satisfy the two-dimensional Toda molecule equation~\cite{AdlerVanMoerbeke1997,Hirota04}. 

\subsection{Hilbert--Schmidt ensemble}\label{sec:HS} 
\begin{proposition}\label{prop1}
The average spectral moment~(\ref{eq:spM}) of the Hilbert--Schmidt ensemble~(\ref{eq:LUE}) satisfies the difference equation
\begin{equation}\label{eq:LUEMdiff}
M_{m+1}(s)+M_{m-1}(s)-\left(\frac{s(s+1)}{m(m+\alpha)}+2\right)M_m(s)=0.
\end{equation}
\end{proposition}

\begin{proof}
In the tau function~(\ref{eq:tau-integral}), we take $\rho({\bf x})$ to be the density~(\ref{eq:LUE}). By Andr\'eief's integration formula~\cite{Andreief1883, Fo18}, the tau function is expressed as a Hankel determinant
\begin{equation}\label{eq:Hd}
\tau_m(t,q;s)=\det\!\left(\mu_{i+j}(t,q;s)\right)_{i,j=0}^{m-1},
\end{equation}
where
\begin{equation}
\mu_k(t,q;s)=\int_0^\infty x^{k+\alpha}\e^{-(1-t)x+qx^s}\dd x.
\end{equation}
It is well known that Hankel determinant of the form~(\ref{eq:Hd}) satisfies Toda lattice equation~\cite{UenoTakasaki1984,AdlerVanMoerbeke1995} of a single time \(t\) as\footnote{In what follows, arguments of tau functions are dropped whenever clear from the context.}
\begin{equation}\label{eq:toda}
\frac{\partial^2}{\partial t^2}\ln\tau_m=\frac{\tau_{m+1}\tau_{m-1}}{\tau_m^2}.
\end{equation}
Taking the derivative of \eqref{eq:toda} with respect to \(q\) and keeping in mind the definition~(\ref{eq:spt}), one has
\begin{equation}\label{eq:linearized-toda}
\frac{\partial^2}{\partial t^2}M_m(s;t)=R_m(t)\left(M_{m+1}(s;t)+M_{m-1}(s;t)-2M_m(s;t)\right),
\end{equation}
where the term
\begin{equation}\label{eq:Rdef}
R_m(t)=\frac{\tau_{m+1}(t,0;s)\tau_{m-1}(t,0;s)}{\tau_m^2(t,0;s)}
\end{equation}
is independent of \(s\) by the definition~(\ref{eq:Hd}). The change of variables 
\begin{equation} 
y_i=(1-t)x_i 
\end{equation} 
in the Hankel determinant~(\ref{eq:Hd}) gives 
\begin{equation}\label{eq:tau-scaling} 
\tau_m(t,q;s) = (1-t)^{-m(m+\alpha)} \tau_m\!\left(0,q(1-t)^{-s};s\right). 
\end{equation}
Setting $q=0$ in (\ref{eq:tau-scaling}) leads to a factorization of Toda time $t$ as
\begin{equation}\label{eq:taup0}
\tau_m(t,0;s)=(1-t)^{-m(m+\alpha)}\tau_m(0,0;s),
\end{equation}
where 
\begin{equation}\label{eq:tau00}
\tau_m(0,0;s)=\prod_{i=1}^m\Gamma(i)\Gamma(\alpha+i)
\end{equation}
is read off from the normalization constant~(\ref{eq:LUEZ}). Consequently, $R_m(t)$ in~(\ref{eq:Rdef}) is simplified to
\begin{equation}\label{eq:R}
R_m(t)=\frac{m(m+\alpha)}{(1-t)^2}.
\end{equation}
Inserting (\ref{eq:tau-scaling}) into~(\ref{eq:spt}) and using the definition~(\ref{eq:spM}), one has
\begin{equation}\label{eq:M-scaling}
M_m(s;t)=(1-t)^{-s}M_m(s)
\end{equation}
and therefore
\begin{equation}\label{eq:mt2}
\left.\frac{\partial^2}{\partial t^2}M_m(s;t)\right\rvert_{t=0}=s(s+1)M_m(s).
\end{equation}
Note that the complete factorization of Toda time $t$ in (\ref{eq:M-scaling}) is a consequence of the homogeneity of densities of the type~\eqref{eq:LUE} under trace rescaling that allows the differentiation with respect to $t$ being trivially performed. The factorizations also extend to Toda deformations with multiple sources (\ref{eq:tau-discussion}) as shown in Lemma~\ref{lemma1} in Appendix~\ref{app1}. Finally, substituting~\eqref{eq:R} and~\eqref{eq:mt2} into~\eqref{eq:linearized-toda} at $t=0$ leads the claimed difference equation~\eqref{eq:LUEMdiff}, which completes the proof of Proposition~\ref{prop1}.
\end{proof}

\begin{remark}
The difference equation~(\ref{eq:LUEMdiff}) was previously obtained in~\cite{CundenMezzadriOConnellSimm} by first identifying a terminating hypergeometric function \({}_3F_2\) representing the spectral moments as Hahn polynomials in $m$, and then applying a known second-order difference equation satisfied by Hahn polynomials.
\end{remark}

Solving special cases of the linear difference equation~(\ref{eq:LUEMdiff}) in Proposition~\ref{prop1} leads to average entropy formulas as summarized in the following corollary.
\begin{corollary}\label{coro1}
Averages of von Neumann entropy~(\ref{eq:von-neumann-entropy}) and quantum purity~(\ref{eq:purity}) of Hilbert--Schmidt ensemble~(\ref{eq:LUE}) satisfy the difference equation
\begin{equation}\label{eq:LUEKdiff}
{M}_{m+1}+{M}_{m-1}-\left(\frac{2}{m(m+\alpha)}+2\right){M}_m=3
\end{equation}
with the solution recovering the formula
\begin{equation}\label{eq:HSentropy} 
\mathbb{E}[S] = \psi_0\!\left(m n+1\right) - \psi_0(n+1) - \frac{m-1}{2 n}
\end{equation} 
and the difference equation
\begin{equation}
M_{m+1}(2)+M_{m-1}(2)-\left(\frac{6}{m(m+\alpha)}+2\right)M_m(2)=0
\label{eq:LUEPdiff}
\end{equation}
with the solution recovering the formula
\begin{equation}\label{eq:HS-average-purity}
\mathbb{E}[P]=\frac{m+n}{mn+1},
\end{equation}
respectively.
\end{corollary}

\begin{proof}
Differentiating~\eqref{eq:LUEMdiff} with respect to \(s\) before setting \(s=1\), and using the definition~(\ref{eq:unconstrained-entropy}), gives
\begin{equation}
{M}_{m+1}+{M}_{m-1}-\left(\frac{2}{m(m+\alpha)}+2\right){M}_m-\frac{3}{m(m+\alpha)}M_m(1)=0.
\end{equation}
The difference equation~\eqref{eq:LUEKdiff} is then established upon substituting 
\begin{equation} 
M_m(1)=\mathbb{E}\!\left[R_1\right]=m(m+\alpha), 
\end{equation} 
which is derived by using~(\ref{eq:spM}) and the fact that the trace $R_1$ of the ensemble~(\ref{eq:LUE}) follows a gamma density~\cite{HW26},
\begin{equation}\label{eq:trace-density} 
f\!\left(R_1\right) = \frac{1}{\Gamma\!\left(m(m+\alpha)\right)}R_1^{m(m+\alpha)-1}\e^{-R_1}. 
\end{equation}

The required initial conditions to solve (\ref{eq:LUEKdiff}) are obtained by using the definition~(\ref{eq:unconstrained-entropy}) as
\begin{eqnarray}
{M}_0 &=& 0 \\
M_1 &=& \frac{1}{\Gamma(\alpha+1)}\int_{0}^{\infty}x^{\alpha+1}\e^{-x}\ln x\,\dd x =(\alpha+1)\psi_0(\alpha+2), \label{eq:M1}
\end{eqnarray}
where the digamma function \(\psi_0\) is the polygamma function $\psi_{i}(x)$,
\begin{equation}\label{eq:polygamma}
\psi_{i}(x)=\frac{\dd^{i+1}}{\dd x^{i+1}}\ln\Gamma(x)=\frac{\dd^{i}}{\dd x^{i}}\psi_{0}(x)
\end{equation}
of order $i=0$. In solving~\eqref{eq:LUEKdiff}, we notice that it is a self-adjoint second-order difference equation, which can be reduced to a first-order one by factoring out the homogeneous solution~\cite{KP01}. Specifically, $M_m(1)$ satisfies the homogeneous version of~\eqref{eq:LUEKdiff},
\begin{equation}\label{eq:c1h}
{M}_{m+1}(1)+{M}_{m-1}(1)-\left(\frac{2}{m(m+\alpha)}+2\right){M}_m(1)=0,
\end{equation}
as seen by setting \(s=1\) in~(\ref{eq:LUEMdiff}).
Due to this, introducing in~\eqref{eq:LUEKdiff} $C_m$ in place of $M_m$ by the change of variable
\begin{equation}
M_m=M_m(1)C_m
\end{equation}
leads to
\begin{equation}\label{eq:LUEC}
M_{m+1}(1)(C_{m+1}-C_m)-M_{m-1}(1)(C_m-C_{m-1})=3.
\end{equation}
Multiplying~(\ref{eq:LUEC}) by \(M_m(1)\) transforms it to the telescoping form
\begin{equation}\label{eq:LUEY}
Y_{m+1}-Y_m=3M_m(1),
\end{equation}
where 
\begin{equation}
Y_m=M_{m-1}(1)M_m(1)(C_m-C_{m-1}).
\end{equation}
Direct summation of (\ref{eq:LUEY}) now yields
\begin{equation}
C_m-C_{m-1}=\frac{2m+3\alpha-1}{2(m+\alpha-1)(m+\alpha)}.
\label{eq:D-solution}
\end{equation}
Summing once more and using the property of digamma function~\cite{Brychkov08}, valid for $x>0$,
\begin{equation}\label{eq:psi0z}
\psi_0(x+1)=\psi_0(x)+\frac{1}{x}, 
\end{equation}
we obtain
\begin{eqnarray}
C_m=\psi_0(m+\alpha+1)+\frac{m-1}{2(m+\alpha)}, \label{eq:C-solution}
\end{eqnarray}
and therefore
\begin{equation}
{M}_m=m(m+\alpha)\psi_0(m+\alpha+1)+\frac{m(m-1)}{2}.
\label{eq:kappa-solution}
\end{equation}
Inserting~(\ref{eq:kappa-solution}) into the moment relation between the two ensembles~(\ref{eq:HS}) and~(\ref{eq:LUE}), see for example~\cite{HW26},
\begin{equation}
\mathbb{E}[S]=\psi_0\!\left(M_m(1)+1\right)-\frac{{M}_m}{M_m(1)},\label{eq:kappa-entropy-conversion}
\end{equation}
and together with the definition~(\ref{eq:alpha}) we recover the formula of average von Neumann entropy~(\ref{eq:HSentropy}) obtained in~\cite{Page93,Ruiz95,Se96}.

For quantum purity, the difference equation (\ref{eq:LUEPdiff}) follows directly by setting $s=2$ in (\ref{eq:LUEMdiff}). Unlike the inhomogeneous equation~\eqref{eq:LUEKdiff}, this homogeneous equation can be solved by seeking a polynomial solution~\cite{KP01}, 
\begin{equation} 
M_m(2)=m(m+\alpha)(2m+\alpha). 
\end{equation}
Applying the moment conversion~\cite{LW21}, 
\begin{equation}\label{eq:HS-purity-conversion}
\mathbb{E}[P]=\frac{M_m(2)}{M_m(1)\bigl(M_m(1)+1\bigr)},
\end{equation}
recovers the formula of average purity~(\ref{eq:HS-average-purity}) obtained in~\cite{LW21,FN2026}. This completes the proof of Corollary~\ref{coro1}.
\end{proof}

\subsection{Bures--Hall ensemble}\label{sec:BH}
\begin{proposition}\label{prop2}
The average spectral moment~(\ref{eq:spM}) of the Bures--Hall ensemble~(\ref{eq:UBH}) satisfies the difference equation
\begin{equation}\label{eq:M-recurrenceBH}
\left(\!\frac{s}{2(m+\alpha)}+1\right)\!M_{m+1}(s)+
\left(1\!-\!\frac{s}{2(m+\alpha)}\right)\!M_{m-1}(s)-
\left(\!\frac{2s(s+1)}{m(m+2\alpha)}+2\right)\!M_m(s)=0.
\end{equation}
\end{proposition}
\begin{proof}
In the tau function~(\ref{eq:tau-integral}), we take $\rho({\bf x})$ to be the density~(\ref{eq:UBH}) so that
\begin{equation}\label{eq:bhtau-definition}
\tau_m(t,q;s)=
\frac{1}{m!}\int_{\R_+^m}
\prod_{1\le i<j\le m}
\frac{(x_i-x_j)^2}{x_i+x_j}
\prod_{i=1}^{m}
x_i^{\alpha-\frac{1}{2}}
\e^{-x_i+t x_i+q x_i^s}\,\dd x_i.
\end{equation}
We consider the signed tau function 
\begin{equation}\label{eq:tautilde}
\Tilde{\tau}_m=(-1)^{\lfloor\frac{m}{2}\rfloor}\tau_m,
\end{equation}
where $\lfloor a \rfloor$ denotes the floor of $a$. By applying the de Bruijn's formulae~\cite{deBruijn1955} and Schur's Pfaffian identities~\cite{Schur1911}, the signed tau function~(\ref{eq:tautilde}) is expressed in terms of Pfaffians~\cite{Chang18} as 
\begin{eqnarray}
\Tilde{\tau}_{2m} &=& \operatorname{Pf}(0,1,\cdots,2m-1)\label{eq:Pfa} \\
\Tilde{\tau}_{2m+1} &=& \operatorname{Pf}\left(d_0,0,1,\cdots,2m\right)\label{eq:Pfa2}
\end{eqnarray}
with Pfaffian entries satisfying
\begin{eqnarray}
\operatorname{Pf}(i,j) &=& \mu_{i, j} \\
\operatorname{Pf}\left(d_0, i\right) &=& \beta_i \\
\frac{\partial}{\partial t} \operatorname{Pf}\left(d_0, i\right) &=& \operatorname{Pf}\left(d_0, i+1\right)= \operatorname{Pf}\left(d_1, i\right) \\
\frac{\partial}{\partial t} \operatorname{Pf}(i, j) &=& \operatorname{Pf}(i+1, j)+\operatorname{Pf}(i, j+1)=\operatorname{Pf}\left(d_0, d_1, i, j\right) \\
\operatorname{Pf}\left(d_0,d_1\right) &=& 0,\label{eq:Pfd0d1}
\end{eqnarray}
where 
\begin{eqnarray}
\mu_{i, j} &=& \iint_{\mathbb{R}_{+}^2} x^{i+\alpha-\frac{1}{2}} y^{j+\alpha-\frac{1}{2}}\frac{x-y}{x+y}\e^{-(1-t)\left(x+y\right)+q\left(x^s+y^s\right)} \dd x \dd y \\
\beta_i &=& \int_{\mathbb{R}_{+}} x^{i+\alpha-\frac{1}{2}} \e^{-(1-t)x+qx^s} \dd x.\label{eq:Pfabeta}
\end{eqnarray}
It is known that the Pfaffian of the form~(\ref{eq:Pfa})--(\ref{eq:Pfabeta}) satisfies the 1+1 B-Toda lattice equation~\cite{Date1982,Jimbo1983,HuLi17,Chang18} of a single time $t$ as
\begin{equation}\label{eq:bilinear-B-Toda}
\Tilde{\tau}_m\,\frac{\partial^2}{\partial t^2}\Tilde{\tau}_m-
\left(\frac{\partial}{\partial t} \Tilde{\tau}_m\right)^2
=\Tilde{\tau}_{m+1}\frac{\partial}{\partial t}\Tilde{\tau}_{m-1}
-\Tilde{\tau}_{m-1}\frac{\partial}{\partial t}\Tilde{\tau}_{m+1}.
\end{equation}
If we now consider log-derivatives of~(\ref{eq:tautilde}), this in turn gives the lattice equation of log-derivatives of the tau function~(\ref{eq:bhtau-definition}) as
\begin{equation}\label{eq:log-B-Toda}
\frac{\partial^2}{\partial t^2}\ln\tau_m=U_m(t,q;s)
\left(\frac{\partial}{\partial t}\ln\tau_{m-1}-\frac{\partial}{\partial t}\ln\tau_{m+1}\right),
\end{equation}
where we denote
\begin{equation}\label{eq:U-definition}
U_m(t,q;s)=-\frac{\tau_{m+1}\tau_{m-1}}{\tau_m^2}
\end{equation}
and the minus sign on its right hand side is a consequence of the definition~(\ref{eq:tautilde}) as
\begin{equation}
\frac{(-1)^{\lfloor(m+1)/2\rfloor+\lfloor(m-1)/2\rfloor}}{(-1)^{2\lfloor m/2\rfloor}}=-1.
\end{equation}
The change of variables
\begin{equation}\label{eq:chagofvariaxy}
y_i=(1-t) x_i
\end{equation}
in (\ref{eq:bhtau-definition}) for the case $q=0$ leads to a complete factorization of $t$, cf.~(\ref{eq:M-scaling}), as
\begin{equation}\label{eq:Z-scaling}
\tau_m(t,0;s)=(1-t)^{-\frac{m(m+2\alpha)}{2}}\tau_m(0,0;s).
\end{equation}
We then have
\begin{equation}
\ln\tau_m(t,0;s)=\ln\tau_m(0,0;s)-\frac{m(m+2\alpha)}{2}\ln(1-t),
\end{equation}
whose partial derivatives with respect to $t$ give
\begin{eqnarray}
\frac{\partial}{\partial t}\ln\tau_m(t,0;s)&=&\frac{m(m+2\alpha)}{2(1-t)} \label{eq:F-derivative1} \\
\frac{\partial^2}{\partial t^2}\ln\tau_m(t,0;s)&=&\frac{m(m+2\alpha)}{2(1-t)^2}.\label{eq:F-derivative2}
\end{eqnarray}
By~(\ref{eq:F-derivative1}), we have 
\begin{eqnarray}
\frac{\partial}{\partial t}\ln \tau_{m-1}(t,0;s)-\frac{\partial}{\partial t}\ln \tau_{m+1}(t,0;s)
=-\frac{2(m+\alpha)}{1-t},
\label{eq:F-neighbor-difference}
\end{eqnarray} 
where substituting \eqref{eq:F-derivative2} and \eqref{eq:F-neighbor-difference} into \eqref{eq:log-B-Toda} at $q=0$ gives
\begin{equation}\label{eq:U-explicit}
U_m(t,0;s)=-\frac{m(m+2\alpha)}{4(m+\alpha)(1-t)}.
\end{equation} 
Differentiating \eqref{eq:U-definition} with respect to \(q\) before setting \(q=0\) and using the definition~(\ref{eq:spt}), one has
\begin{equation}\label{eq:q-derivative-U}
\left.\frac{\partial}{\partial q}U_m(t,q;s)\right\rvert_{q=0}=U_m(t,0;s)\left(M_{m+1}(s;t)+M_{m-1}(s;t)-2M_m(s;t)\right).
\end{equation}
Consequently, the lattice equation~(\ref{eq:log-B-Toda}) at $q=0$ becomes
\begin{eqnarray}\label{eq:linearized-B-Toda}
\!\!\!\!\!\!\frac{\partial^2}{\partial t^2}M_m(s;t)&=&U_m(t,0;s)\bigg(\!\left(M_{m+1}(s;t)+M_{m-1}(s;t)-2M_m(s;t)\right) \nonumber \\
&&\!\!\!\!\!\!\!\!\!\!\!\!\!\!\!\!\!\!\!\!\!\!\!\!\!\!\!\!\!\!\!\!\!\!\!\!\times\left(\frac{\partial}{\partial t}\ln \tau_{m-1}(t,0;s)-\frac{\partial}{\partial t}\ln\tau_{m+1}(t,0;s)\right)+\frac{\partial}{\partial t}M_{m-1}(s;t)-\frac{\partial}{\partial t}M_{m+1}(s;t)\bigg).
\end{eqnarray}
By the definition~\eqref{eq:spt}, the $k$-th derivative of~$M_m(s;t)$ with respect to $t$, evaluated at $t=0$, generates a joint cumulant of $R_s$ and $k$ copies of $R_1$. Hence, invoking~(\ref{eq:general-cumulant-scaling}) in Lemma~\ref{lemma1} for the cases $l=1, n=1$ and $l=1, n=2$, we have 
\begin{equation}\label{eq:bhM-derivative}
\left.\frac{\partial}{\partial t}M_m(s;t)\right\rvert_{t=0}=s\mathbb{E}\!\left[R_s\right]=sM_m(s) 
\end{equation}
and 
\begin{equation}\label{eq:M-derivatives}  
\left.\frac{\partial^2}{\partial t^2}M_m(s;t)\right\rvert_{t=0}=s(s+1)\mathbb{E}\!\left[R_s\right]=s(s+1)M_m(s),
\end{equation}
respectively. Finally, substituting \eqref{eq:F-neighbor-difference}, \eqref{eq:U-explicit}, \eqref{eq:bhM-derivative}, and \eqref{eq:M-derivatives} into \eqref{eq:linearized-B-Toda} yields the difference equation (\ref{eq:M-recurrenceBH}), which completes the proof of Proposition~\ref{prop2}.
\end{proof}

Solving special cases of the linear difference equation~(\ref{eq:M-recurrenceBH}) in Proposition~\ref{prop2} leads to average entropy formulas as presented in the corollary below.
\begin{corollary}\label{cor2}
Averages of von Neumann entropy~(\ref{eq:von-neumann-entropy}) and quantum purity~(\ref{eq:purity}) of Bures--Hall ensemble~(\ref{eq:UBH}) satisfy the difference equation
\begin{equation}\label{eq:kappa-recurrenceBH}
\!\!\left(\frac{1}{2(m+\alpha)}+1\right){M}_{m+1}
+\left(1-\frac{1}{2(m+\alpha)}\right){M}_{m-1}
-\left(\frac{4}{m(m+2\alpha)}+2\right){M}_m=2
\end{equation}
with the solution recovering the formula
\begin{equation}
\mathbb{E}[S] = \psi_0\!\left(m n-\frac{m^2}{2}+1\right)-\psi_0\!\left(n+\frac{1}{2}\right),
\label{eq:BH-entropy-solution} 
\end{equation}
and the difference equation
\begin{equation}
\left(\!\frac{2}{2(m+\alpha)}+1\right) \!M_{m+1}(2)+\left(1\!-\!\frac{2}{2(m+\alpha)}\right)\!M_{m-1}(2)-\left(\frac{12}{m(m+2\alpha)}+2\right)\!M_m(2)=0
\label{eq:uBHPdiff}
\end{equation}
with the solution recovering the formula
\begin{eqnarray}\label{eq:bhpurity}
\mathbb{E}[P]=\frac{m^2-2 m n-4 n^2-1}{2 n\left(m^2-2 m n-2\right)},
\end{eqnarray}
respectively.
\end{corollary}

\begin{proof}
Differentiating \eqref{eq:M-recurrenceBH} with respect to \(s\) before setting \(s=1\), and using the definition~(\ref{eq:unconstrained-entropy}), gives
\begin{eqnarray}\label{eq:kappa-recurrenceBHorigin}
&&\left(\frac{1}{2(m+\alpha)}+1\right){M}_{m+1}
+\left(1-\frac{1}{2(m+\alpha)}\right){M}_{m-1}
-\left(\frac{4}{m(m+2\alpha)}+2\right){M}_m\nonumber\\
&=&\frac{6 M_m(1)}{m(m+2\alpha)}-\frac{M_{m+1}(1)-M_{m-1}(1)}{2(m+\alpha)}.
\end{eqnarray}
The difference equation~(\ref{eq:kappa-recurrenceBH}) is then established upon substituting
\begin{equation}\label{eq:first-momentbh}
M_m(1)=\mathbb{E}\!\left[R_1\right]=\frac{m(m+2\alpha)}{2},
\end{equation}
which is derived by using~(\ref{eq:spM}) and the fact that the trace $R_1$ of the ensemble~(\ref{eq:UBH}) follows a gamma density~\cite{Wei23C},
\begin{equation}\label{eq:trace-densityBH} 
f\!\left(R_1\right) = \frac{1}{\Gamma\!\left(\frac{m(m+2\alpha)}{2}\right)}R_1^{\frac{m(m+2\alpha)}{2}-1}\e^{-R_1}. 
\end{equation}

The required initial conditions to solve~(\ref{eq:kappa-recurrenceBH}) are obtained by using~(\ref{eq:unconstrained-entropy}), cf.~(\ref{eq:M1}), as
\begin{eqnarray}
{M}_0 &=& 0 \\
{M}_1 &=& \left(\alpha+\frac{1}{2}\right)\psi_0\!\left(\alpha+\frac{3}{2}\right). \label{eq:kappa-initial}
\end{eqnarray}
Similar to (\ref{eq:LUEKdiff}), one identifies a self-adjoint second-order difference equation~\cite{KP01} after multiplying by $m+\alpha$ on both sides of (\ref{eq:kappa-recurrenceBH}),
\begin{equation*}\label{eq:kappa-recurrenceBH1}
\left(m+\alpha+\frac{1}{2}\right)M_{m+1}+\left(m+\alpha-\frac{1}{2}\right)M_{m-1}-\left(2(m+\alpha)+\frac{4(m+\alpha)}{m(m+2\alpha)}\right)M_m=2(m+\alpha),
\end{equation*}
which reduces to a first-order one by factoring out the homogeneous solution. Specifically, since $M_m(1)$ satisfies the homogeneous version of~(\ref{eq:kappa-recurrenceBH}),
\begin{equation}\label{eq:M-recurrenceBH-s1}
\left(m+\alpha+\frac{1}{2}\right)M_{m+1}(1)+\left(m+\alpha-\frac{1}{2}\right)M_{m-1}(1)-\left(2(m+\alpha)+\frac{4(m+\alpha)}{m(m+2\alpha)}\right)M_m(1)=0,
\end{equation}
as seen by setting $s=1$ in (\ref{eq:M-recurrenceBH}). Due to this, introducing in~(\ref{eq:kappa-recurrenceBH}) $C_m$ in place of $M_m$ by the change of variable 
\begin{equation}
{M}_m=M_m(1) C_m
\label{eq:C-substitutionBH}
\end{equation}
leads to 
\begin{equation}\label{eq:kappa-recurrenceBH2}
\left(m+\alpha+\frac{1}{2}\right)M_{m+1}(1)
(C_{m+1}-C_m)
-\left(m+\alpha-\frac{1}{2}\right)M_{m-1}(1)
(C_m-C_{m-1})
=2(m+\alpha).
\end{equation}
Multiplying (\ref{eq:kappa-recurrenceBH2}) by $M_m(1)$ transforms it to the telescoping form
\begin{equation}\label{eq:telebh}
Y_{m+1}-Y_m=2(m+\alpha) M_m(1),
\end{equation}
where
\begin{equation}\label{YmBHdef}
Y_m = \left(m+\alpha-\frac{1}{2} \right)M_{m-1}(1)M_m(1)(C_m-C_{m-1}).
\end{equation}
Direct summation of (\ref{eq:telebh}) yields
\begin{eqnarray}
C_m-C_{m-1}=\frac{2}{2m+2\alpha-1}.
\end{eqnarray}
Summing once more and using the property~(\ref{eq:psi0z}), we obtain
\begin{eqnarray}
 C_m&=&\psi_0\!\left(m+\alpha+\frac{1}{2}\right),\label{eq:C-solutionBH}
\end{eqnarray}
and hence
\begin{equation}
{M}_m=M_m(1)\psi_0\!\left(m+\alpha+\frac{1}{2}\right).\label{eq:kappa-solutionBH}
\end{equation}
Inserting~(\ref{eq:first-momentbh}) and (\ref{eq:kappa-solutionBH}) into the moment relation (\ref{eq:kappa-entropy-conversion}) between two ensembles (\ref{eq:BH}) and (\ref{eq:UBH}), see for example~\cite{Wei20BHA,WW23}, and together with the definition~\eqref{eq:alpha}, we recover the average von Neumann entropy~\eqref{eq:BH-entropy-solution} obtained in~\cite{Sarkar2019, Wei20BHA, WHW26}.

For the quantum purity, the difference equation~(\ref{eq:uBHPdiff}) is obtained by setting $s=2$ in~(\ref{eq:M-recurrenceBH}). Similar to~(\ref{eq:LUEPdiff}), the homogeneous equation~(\ref{eq:uBHPdiff}) can be solved by seeking a polynomial solution~\cite{KP01} as
\begin{equation}
M_m(2)=\frac{m(m+2\alpha)\left(5m^2+10m\alpha+4\alpha^2+1\right)}{8(m+\alpha)}.
\end{equation}
Applying the moment conversion~(\ref{eq:HS-purity-conversion}), we recover the formula of average purity~(\ref{eq:bhpurity}) obtained in~\cite{Sarkar2019,Wei20BHA,LW21}. This completes the proof of Corollary~\ref{cor2}.
\end{proof}

\subsection{Muttalib--Borodin ensemble}\label{sec:MB}
Our result here for the difference equation of spectral moments is in terms of its covariance
\begin{equation}\label{def:cov}
\kappa\!\left(R_{s_1}, R_{s_2}\right) =
\mathbb E[R_{s_1} R_{s_2} ] - \mathbb E[R_{s_1}] \mathbb E[R_{s_2}] =
\mathbb E[R_{s_1} R_{s_2} ] - M_m(s_1) M_m(s_2).
\end{equation}

\begin{proposition}\label{prop3}
The spectral moment~(\ref{eq:Rs}) of the Muttalib--Borodin ensemble~(\ref{eq:LMB}) satisfies the relation
\begin{equation}\label{eq:MBd}
(s+\theta)\kappa\!\left(R_s,R_\theta\right)=\frac{\theta m~\!\Gamma(\alpha+1+\theta m)}{\Gamma(\alpha+1+\theta(m-1))}\left(M_{m+1}(s)+M_{m-1}(s)-2M_m(s)\right).
\end{equation}
\end{proposition}

\begin{proof}
In the tau function~(\ref{eq:MB-tau-integral}), we take $\rho({\bf x})$ to be the density~(\ref{eq:LMB}), so that
\begin{equation}\label{eq:tauMB}
\tau_m(t,q,u;s,\theta)=\frac{1}{m!}\int_{\R_+^m}\prod_{1\leq i<j\leq m}(x_i-x_j)\left(x_i^\theta-x_j^\theta\right)\prod_{i=1}^{m}x_i^\alpha\e^{-x_i+tx_i+qx_i^s+u x_i^\theta}\dd x_i.
\end{equation}
The corresponding spectral moment follows from
\begin{equation}
M_m(s)=M_m(s;0,0),
\end{equation}
where
\begin{equation}\label{eq:MB-deformed-moment}
M_m(s;t,u)=\left.\frac{\partial}{\partial q}\ln\tau_m(t,q,u;s,\theta)\right\rvert_{q=0}.
\end{equation}
By Andr\'eief's formula~\cite{Andreief1883, Fo18}, the tau function~(\ref{eq:tauMB}) is written as a bimoment determinant
\begin{equation}
\tau_m(t,q,u;s,\theta)=\det\!\left(\mu_{ij}(t,q,u;s,\theta)\right)_{i,j=0}^{m-1},
\end{equation}
where
\begin{equation}
\mu_{ij}(t,q,u;s,\theta)=\int_0^\infty x^{\alpha+i+\theta j}\e^{-(1-t)x+u x^\theta+q x^s}\dd x .
\end{equation}
The derivatives shift the row and column as
\begin{eqnarray}
\frac{\partial}{\partial t}\mu_{ij}&=&\mu_{i+1,j} \\
\frac{\partial}{\partial u}\mu_{ij}&=&\mu_{i,j+1}, 
\end{eqnarray}
where, together with multi-linear property of determinants, one arrives at
\begin{eqnarray}
\frac{\partial}{\partial t}\tau_m
&=&
\det\!\left(\mu_{ij}\right)_
{\substack{i=0,\ldots,m-2,m\\j=0,\ldots,m-1}}\\
\frac{\partial}{\partial u}\tau_m
&=&
\det\!\left(\mu_{ij}\right)_
{\substack{i=0,\ldots,m-1\\j=0,\ldots,m-2,m}}\\
\frac{\partial^2}{\partial t \partial u}\tau_m
&=&
\det\!\left(\mu_{ij}\right)_
{\substack{i=0,\ldots,m-2,m\\j=0,\ldots,m-2,m}}.
\label{eq:MB-shifted-determinants}  
\end{eqnarray}
The Desnanot--Jacobi determinant identity~\cite{SNV11} now relates the above quantities as
\begin{equation}\label{eq:MB-Jacobi-identity}
\tau_m\,\frac{\partial^2}{\partial t\partial u}\tau_m-\left(\frac{\partial}{\partial t}\tau_m\right)\left(\frac{\partial}{\partial u}\tau_m\right)=\tau_{m+1}\tau_{m-1},
\end{equation}
which can be rewritten as the two-dimensional Toda molecule equation~\cite{AdlerVanMoerbeke1997,Hirota04},
\begin{equation}\label{eq:2D-Toda}
\frac{\partial^2}{\partial t \partial u}\ln\tau_m=
\frac{\tau_{m+1}\tau_{m-1}}{\tau_m^2}.
\end{equation}
Differentiating \eqref{eq:2D-Toda} with respect to \(q\) before setting \(q=0\) gives
\begin{align}
\frac{\partial^2}{\partial t\partial u} M_m(s;t,u)
={}&R_m(t,u)\Bigl(M_{m+1}(s;t,u)+M_{m-1}(s;t,u)-2M_m(s;t,u)\Bigr),
\label{eq:MB-linearized-toda}
\end{align}
where
\begin{equation}
R_m(t,u)=\left.\frac{\tau_{m+1}(t,q,u;s,\theta)\tau_{m-1}(t,q,u;s,\theta)}{\tau_m^2(t,q,u;s,\theta)}\right\rvert_{q=0}.
\label{eq:MB-R-ratio}
\end{equation}
From the normalization constant~(\ref{eq:zMB}), we must have
\begin{equation}\label{eq:MB-partition-origin}
\tau_m(0,0,0;s,\theta)=\theta^{m(m-1)/2}\prod_{i=1}^{m}\Gamma(i)\,\Gamma\!\left(\alpha+1+\theta(i-1)\right),
\end{equation}
inserting which into (\ref{eq:MB-R-ratio}) gives 
\begin{equation}\label{eq:MB-toda-coefficient}
R_m(0,0)=\frac{\theta m~\!\Gamma(\alpha+1+\theta m)}{\Gamma(\alpha+1+\theta(m-1))}.
\end{equation}
Applying Lemma~\ref{lemma1} in Appendix~\ref{app1} for the case \(n=1\), \(l=2\), \(s_1=s\), \(s_2=\theta\), together with the definition~(\ref{eq:MB-deformed-moment}), one obtains
\begin{equation}\label{eq:MB-covariance-derivative}
\left.\frac{\partial^2}{\partial t\partial u}M_m(s;t,u)\right\rvert_{t=u=0}=(s+\theta)\kappa\!\left(R_s,R_\theta\right).
\end{equation}
Evaluating \eqref{eq:MB-linearized-toda} at \(t=u=0\) before inserting \eqref{eq:MB-toda-coefficient} and \eqref{eq:MB-covariance-derivative} into the resulting expression, we arrive at the claimed relation \eqref{eq:MBd}. This completes the proof of Proposition~\ref{prop3}.
\end{proof}

For the special case $\theta=1$ in Proposition~\ref{prop3}, applying the cumulant identity \eqref{eq:general-cumulant-scaling} to the left-hand side of (\ref{eq:MBd}) directly recovers the difference equation~(\ref{eq:LUEMdiff}) of the Hilbert--Schmidt ensemble in Proposition~\ref{prop1}. 

For the special case $\theta\to0$ in Proposition~\ref{prop3}, the limit establishes the difference equation of the Bogoliubov--Kubo--Mori ensemble, which in turn recovers average entropy formulas of the ensemble. The results are summarized in the corollary below, where, unlike the Hilbert–Schmidt ensemble of $\theta=1$, the limiting procedure is essential in the sense that there is no direct derivation available.

\begin{corollary}\label{co:BKM}
The average spectral moment (\ref{eq:spM}) of the Bogoliubov--Kubo--Mori ensemble~(\ref{eq:uBKM}) satisfies the difference equation
\begin{equation}\label{eq:BKMdd}
s\frac{\partial}{\partial\alpha}M_m(s)=m\left(M_{m+1}(s)+M_{m-1}(s)-2M_m(s)\right)
\end{equation}
with the solution recovering the formula of average entropy
\begin{equation}\label{eq:BKMS}
\mathbb{E}[S]=\psi_{0}(\beta+1)-\frac{1}{\beta}
\left(\sum_{i=0}^{m-1}\binom{m}{i+1}\frac{1}{i!}\left(\alpha+\frac{m+1}{i+2}\right)\psi_{i}(\alpha+1)+\frac{1}{2}m(m+1)\right),
\end{equation}
where $\beta=m\alpha+m(m+1)/2$.
\end{corollary}

\begin{proof}
In the limit \(\theta\to0\), the Muttalib--Borodin ensemble~\eqref{eq:LMB} degenerates to the Bogoliubov--Kubo--Mori ensemble~\eqref{eq:uBKM}. This follows from resolving the indeterminacy
\begin{equation}
\lim_{\theta\to 0}\frac{x_i^\theta-x_j^\theta}{\theta}=\ln x_i-\ln x_j
\end{equation}
as a consequence of the $\theta\to0$ expansion for $x_i>0$,
\begin{equation}\label{eq:xthetalim}
x_i^\theta=1+\theta\ln x_i+\mathcal{O}\!\left(\theta^2\right).
\end{equation}
Consequently, the spectral moments of the ensemble (\ref{eq:LMB}) reduce those of the ensemble (\ref{eq:uBKM}) as \(\theta\to0\).

Summing (\ref{eq:xthetalim}) over \(i=1,\ldots,m\) gives an expansion of spectral moments $R_\theta$ in the limit~\(\theta \to 0\) for a fixed $m$ as
\begin{equation}\label{eq:Rthetalim}
R_\theta=m+\theta T_0+\mathcal{O}\!\left(\theta^2\right)
\end{equation}
with $T_0$ denoting the random variable
\begin{equation}\label{eq:T0}
T_0=\sum_{i=1}^{m}\ln x_i,
\end{equation}
consistent with the notation in~\cite{HW26}. Inserting (\ref{eq:Rthetalim}) into the definition (\ref{def:cov}), one has
\begin{equation} \label{eq:covthetalim}
\kappa(R_s,R_\theta)
=\theta\kappa\!\left(R_s,T_0\right)+\mathcal{O}\!\left(\theta^2\right),
\end{equation}
where, on the other hand, it is known that joint cumulants involving $T_0$ can be generated by $\alpha$ derivatives~\cite{HW26},
\begin{equation}\label{eq:dalpha1}
\kappa\!\left(R_s,T_0\right)=\frac{\partial}{\partial\alpha} M_m(s).
\end{equation}
Finally, inserting~\eqref{eq:covthetalim} into~\eqref{eq:MBd} before taking the limit $\theta\to 0$ establishes the desired difference equation~\eqref{eq:BKMdd}.

We now move on to the average entropy formula~(\ref{eq:BKMS}). A general approach to solve the difference equation of the form~\eqref{eq:BKMdd} is through its the generating function. Here, we present an alternative approach that reveals the underlying structure of the ensemble's spectral moments. Introducing the differential operator
\begin{equation}\label{eq:D}
\mathcal D=-s\frac{\partial}{\partial\alpha},
\end{equation}
we can rewrite~\eqref{eq:BKMdd} as
\begin{equation}\label{eq:BKMddD}
M_{m+1}(s)=\left(-\frac{\mathcal D}{m}+2\right)M_m(s)-M_{m-1}(s).
\end{equation}
From the initial conditions, cf.~(\ref{eq:M1}),
\begin{eqnarray}
M_0(s)&=&0 \\
M_1(s)&=&\frac{1}{\Gamma(\alpha+1)}\int_0^\infty x^{\alpha+s}\e^{-x}\,\dd x =\frac{\Gamma(\alpha+s+1)}{\Gamma(\alpha+1)},
\end{eqnarray}
it follows inductively that
\begin{equation}\label{eq:pMD}
M_m(s)=P_{m-1}(\mathcal D)M_1(s),
\end{equation}
where \(P_m\) denotes a polynomial of degree \(m\). Inserting (\ref{eq:pMD}) into (\ref{eq:BKMddD}) gives
\begin{equation}\label{eq:reD}
\mathcal{D}P_{m-1}(\mathcal{D})=-m P_m(\mathcal{D})+2mP_{m-1}(\mathcal{D})-mP_{m-2}(\mathcal{D})
\end{equation}
with initial conditions computed from~(\ref{eq:BKMddD}) and~(\ref{eq:pMD}) as
\begin{eqnarray}
P_0(\mathcal{D}) &=& 1 \\
P_1(\mathcal{D}) &=& -\mathcal{D} + 2.
\end{eqnarray}

We observe that the relation~(\ref{eq:reD}) is nothing but the recurrence relation of Laguerre polynomials~\cite{Szego},
\begin{equation}\label{eq:recur}
xL_k^{(a)}(x)=-(k+a) L_{k-1}^{(a)}(x)+(2k+1+a)L_k^{(a)}(x)-(k+1)L_{k+1}^{(a)}(x)
\end{equation}
of the specialization $k= m-1$, $a=1$, $x=\mathcal{D}$, namely,
\begin{equation}\label{eq:PL}
P_{m}(\mathcal{D})=L_{m}^{(1)}(\mathcal{D}).
\end{equation}
Substituting~\eqref{eq:PL} into~\eqref{eq:pMD}, we obtain a closed-form expression of spectral moments
\begin{equation}\label{eq:uBKMMs}
M_m(s)=L_{m-1}^{(1)}(\mathcal{D})\frac{\Gamma(\alpha+s+1)}{\Gamma(\alpha+1)}.
\end{equation}
By the definition~\eqref{eq:unconstrained-entropy}, differentiating (\ref{eq:uBKMMs}) with respect to \(s\) and setting \(s=1\), we will show that 
\begin{equation}
M_m=\frac{m(m-1)}{2}+L_{m-1}^{(1)}\!\left(-\frac{\partial}{\partial\alpha}\right)(\alpha+1)\,\psi_0(\alpha+2).
\label{eq:uBKMV}
\end{equation}
The first term in (\ref{eq:uBKMV}) arises from the derivative of $L_{m-1}^{(1)}(\mathcal D)$ with respect to $s$, where we recall that $\mathcal D$ depends on $s$ by definition (\ref{eq:D}). The chain rule together with the identity~\cite{Szego},
\begin{equation}
\frac{\partial}{\partial x}L_m^{(a)}(x)=-L_{m-1}^{(a+1)}(x),
\end{equation} 
gives
\begin{eqnarray}
\left.\frac{\partial}{\partial s}L_{m-1}^{(1)}(\mathcal D)\right\rvert_{s=1}\frac{\Gamma(\alpha+2)}{\Gamma(\alpha+1)}&=&L_{m-2}^{(2)}\!\left(-\frac{\partial} {\partial\alpha}\right)\frac{\partial}{\partial\alpha}(\alpha+1) \nonumber \\ 
&=&L_{m-2}^{(2)}(0)=\frac{m(m-1)}{2},    
\end{eqnarray}
where the last two equalities follow from the definition~\cite{Szego},
\begin{equation}\label{eq:lpfs}
L_{m}^{(a)}(x)=\sum_{i=0}^{m}(-1)^i\binom{m+a}{m-i}\frac{x^i}{i!}.
\end{equation}
The second term in~\eqref{eq:uBKMV} arises from the derivative acting on \(\Gamma(\alpha+s+1)\), evaluated using the definition of the polygamma function~\eqref{eq:polygamma}.

Applying (\ref{eq:psi0z}) and (\ref{eq:lpfs}) to (\ref{eq:uBKMV}) leads to 
\begin{equation}
M_m=\frac{m(m+1)}{2}+\sum_{i=0}^{m-1}\binom{m}{i+1}\frac{1}{i!}\left(\alpha+\frac{m+1}{i+2}\right)\psi_i(\alpha+1),
\label{eq:uBKMV1}
\end{equation}
where we used the definition of polygamma functions~(\ref{eq:polygamma}). Finally, inserting (\ref{eq:uBKMV1}) into the moment relation~\cite{Mi25,SW26},
\begin{equation}
\mathbb{E}[S]=\psi_0(\beta+1)-\frac{M_m}{\beta}
\end{equation}
with $\beta=m\alpha+m(m+1)/2$, we recover the formula of average von Neumann entropy (\ref{eq:BKMS}) recently obtained in~\cite{SW26}. This completes the proof of Corollary~\ref{co:BKM}.
\end{proof}

\begin{remark}
The Muttalib--Borodin ensemble~(\ref{eq:LMB}) is a $\theta$-deformation of the Hilbert--Schmidt ensemble \eqref{eq:LUE}. The analogous $\theta$-deformation of the Bures--Hall ensemble \eqref{eq:UBH}, also obtained by replacing each factor $(x_i - x_j)^2$ by $(x_i - x_j)(x_i^\theta - x_j^\theta)$, is known to exhibit exactly solvable and integrable properties \cite{FL19,CLTY21}. An investigation of a moment recurrence analogous to (\ref{eq:MBd}) remains to be undertaken.
\end{remark}

\section{Relations to known spectral moment recurrences}\label{sec:smr}
In this section, we establish the connection between the difference relations of spectral moments $M_m(s)$ with respect to matrix dimension $m$ derived in Section~\ref{sec:results} and those with respect to power parameter $s$ known in the literature~\cite{Harer86,HT03, Ledoux2004,Ledoux2009ARF,CundenMezzadriOConnellSimm,WHW26}. This is done by deriving mixed difference equations in both $m$ and $s$ for the Hilbert--Schmidt~\eqref{eq:LUE} and Bures--Hall ensembles~\eqref{eq:UBH} in Section~\ref{sec:3.1} and Section~\ref{sec:3.2}, respectively. The additional ingredients are Virasoro constraints in the loop-equation form~\cite{MMMM91,FRW17}.

\subsection{Mixed difference equation of Hilbert--Schmidt ensemble}\label{sec:3.1}
\begin{proposition}\label{prop4}
The average spectral moment~(\ref{eq:spM}) of the Hilbert--Schmidt ensemble~(\ref{eq:LUE}) satisfies the mixed difference equation in $m$ and $s$,
\begin{equation}\label{eq:luemixed}
 m(m+\alpha)\left(M_{m+1}(s)-M_{m-1}(s)\right)=(s+2)M_m(s+1)-s(2m+\alpha)M_m(s),
\end{equation}
combining which with the difference equation~(\ref{eq:LUEMdiff}) in $m$ reproduces the known difference equation in $s$,
\begin{equation}\label{eq:LUE-s-recurrence}
(s+3)M_m(s+2)-(2s+3)(2m+\alpha)M_m(s+1)-s\left((s+1)^2-\alpha^2\right)M_m(s)=0.
\end{equation}
\end{proposition}

\begin{proof}
The Virasoro constraints~\cite{MM90,AdlerVanMoerbeke1995} reveal relations of spectral moments $M_m(s)$ of different power parameters $s$. For a nonnegative integer $r$, our starting point is the integration by parts underlying the Virasoro constraints~\cite{MM90,HH93,AdlerVanMoerbeke1995}, 
\begin{equation}\label{eq:ibp-virasoro}
0=\int_{\R_+^m}\sum_{i=1}^m\frac{\partial}{\partial x_i}\left(x_i^{r+1}R_s\!~\rho({\bf x})\right)\prod_{i=1}^m \dd x_i.
\end{equation}
Here, $s$ is a real-valued parameter so that the subsequent results involving $s$ in the proof are also valid for real \(s\). Consider now the Hilbert--Schmidt ensemble~(\ref{eq:LUE}), taking the logarithmic derivative of the density gives
\begin{equation}\label{eq:3.1drho}
\frac{\partial}{\partial x_i}\rho(\mathbf{x})=\rho(\mathbf{x})\left(2\sum_{j\ne i}\frac{1}{x_i-x_j}+\frac{\alpha}{x_i}-1\right),
\end{equation}
inserting which into~\eqref{eq:ibp-virasoro}, we have
\begin{eqnarray}\label{eq:expanded-virasoro}  
\!\!\!\!\!\!\!\!0=(r+1)\mathbb{E}[R_rR_s]+s\mathbb{E}[R_{s+r}]+\alpha\mathbb{E}[R_rR_s]-\mathbb{E}[R_{r+1}R_s]+2\mathbb{E}\!\left[R_s\sum_{i=1}^{m}\sum_{j\ne i}\frac{x_i^{r+1}}{x_i-x_j}\right].
\end{eqnarray}
Using the symmetry, see for example~\cite{Aomoto87,Forrester},
\begin{equation}
2\sum_{i=1}^m\sum_{j\ne i}\frac{x_i^{r+1}}{x_i-x_j}=\sum_{\ell=0}^{r}R_\ell R_{r-\ell}-(r+1)R_r,
\label{eq:pair-sum}
\end{equation}
one obtains the loop equation~\cite{FRW17,GGR20},
\begin{equation}
\mathbb{E}[R_{r+1}R_s]
=\alpha\mathbb{E}[R_rR_s]+\sum_{\ell=0}^{r}\mathbb{E}[R_\ell R_{r-\ell}R_s]+ s\mathbb{E}[R_{s+r}].
\label{eq:loopr}
\end{equation}
Setting $r=1$ in (\ref{eq:loopr}) gives
\begin{equation}\label{eq:virasoror2}
\mathbb{E}[R_sR_{2}] = (2m+\alpha) \mathbb{E}[R_sR_1]+s \mathbb{E}[R_{s+1}],
\end{equation}
which can be rewritten by using the definitions~(\ref{eq:spM}) and~(\ref{def:cov}) as
\begin{equation}\label{eq:loop1u}
\kappa(R_s,R_2)+M_m(s)M_m(2)=(2m+\alpha)\kappa(R_s,R_1)+(2m+\alpha)M_m(s)M_m(1)+sM_m(s+1).
\end{equation}
To simplify~(\ref{eq:loop1u}), we use Lemma~\ref{lemma1} that gives
\begin{equation}\label{eq:covr}
\kappa(R_s,R_1)=s M_m(s)
\end{equation}
and the required initial conditions  
\begin{eqnarray}
M_m(1)&=&m(m+\alpha) \label{eq:luem1} \\ 
M_m(2)&=&(2m+\alpha)m(m+\alpha) \label{eq:luem2}
\end{eqnarray}
can be, for example, obtained from~\cite{HT03,HW26}, the results simplify (\ref{eq:loop1u}) to
\begin{equation}\label{eq:lueloop1}
\kappa(R_s,R_2)=s\big((2m+\alpha)M_m(s)+M_m(s+1)\big).
\end{equation}

The covariance in~(\ref{eq:lueloop1}) can be eliminated by introducing another flow of the Toda hierarchy~\cite{UenoTakasaki1984,AdlerVanMoerbeke1995,AdlerVanMoerbeke1997, AdlerHorozovVanMoerbeke1999}. Higher flows of Toda equations provide additional independent relations that can be utilized to eliminate undesired quantities. Specifically, we extend the tau function (\ref{eq:tau-integral}) to include $t_2$,
\begin{equation}\label{eq:LUE-extended-tau}
\tau_m(t_1,t_2,q;s)=\frac{Z}{m!}\int_{\R_+^m}\rho({\bf x})\prod_{i=1}^{m}\e^{t_1 x_i+t_2 x_i^2+q x_i^s}\dd x_i,
\end{equation}
which is known to satisfy~\cite{UenoTakasaki1984,AdlerVanMoerbeke1995},
\begin{equation}\label{eq:LUE-t1t2-Toda}
\frac{\partial^2}{\partial t_1\partial t_2}\ln\tau_m=\frac{\tau_{m+1}\tau_{m-1}}{\tau_m^2}\frac{\partial}{\partial t_1}\ln\frac{\tau_{m+1}}{\tau_{m-1}}.
\end{equation}
Note that the equation~(\ref{eq:LUE-t1t2-Toda}) can be written more compactly in the bilinear  form~\cite{AdlerVanMoerbeke19992} as 
\begin{equation}
\frac{1}{2}D_{t_1}D_{t_2}\tau_m\cdot\tau_m=D_{t_1}\tau_{m+1}\cdot\tau_{m-1},
\end{equation}
where $D_x$ denotes the Hirota's derivative~\cite{Hirota04}, 
\begin{eqnarray}\label{eq:hirotaderiv1}
D_x^r f\cdot g&=&\left.\frac{\partial^r}{\partial y^r}f(x+y)g(x-y)\right|_{y=0}\\
D_x^{r_1}D_t^{r_2}f\cdot g&=&\left.\frac{\partial^{r_1}}{\partial y^{r_1}}
\frac{\partial^{r_2}}{\partial s^{r_2}}f(x+y,t+s)g(x-y,t-s)\right|_{y=0,\,s=0}.\label{eq:hirotaderiv2}
\end{eqnarray}

Now taking the derivative of (\ref{eq:LUE-t1t2-Toda}) with respect to \(q\) and setting \(t_1=t_2=q=0\), by the definition~(\ref{eq:spM}) and the fact that 
\begin{eqnarray}
\left.\frac{\partial^2}{\partial t_1\partial q}\ln\tau_m(t_1,t_2,q;s)\right\rvert_{t_1=t_2=q=0}&=&\kappa(R_1,R_s) \\
\left.\frac{\partial^3}{\partial t_1\partial t_2\partial q}\ln\tau_m(t_1,t_2,q;s)\right\rvert_{t_1=t_2=q=0}&=&\kappa(R_1,R_2,R_s), \label{eq:LUEkappads}
\end{eqnarray}
we obtain 
\begin{eqnarray}
\kappa(R_1,R_2,R_s)&=&M_m(1)\left(M_{m+1}(s)+M_{m-1}(s)-2M_m(s)\right)\left(M_{m+1}(1)-M_{m-1}(1)\right)\nonumber\\
&&+\!~M_m(1)\left(\kappa_{(m+1)}(R_1,R_s)-\kappa_{(m-1)}(R_1,R_s)\right). \label{eq:LUEt1t2-intermediate}
\end{eqnarray}
The above results are written using the definition of joint cumulant of $l$ random variables \(X_1,\ldots,X_l\) as~\cite{PT11},
\begin{equation}\label{eq:joint-cumulant-def}
\kappa(X_1,\ldots,X_l)=\left.\frac{\partial^l}{\partial u_1\cdots\partial u_l}\ln\mathbb{E}\!\left[\e^{\sum_{i=1}^l u_iX_i}\right]\right\rvert_{{\bf u}={\bf 0}}
\end{equation}
that extends the definition of covariance~\eqref{def:cov}, where the subscript of joint cumulant $\kappa$ in parentheses in~(\ref{eq:LUEt1t2-intermediate}) denotes the matrix dimension when it differs from $m$. The equation~(\ref{eq:LUEt1t2-intermediate}) is simplified, by invoking Lemma~\ref{lemma1},
\begin{eqnarray}
\kappa(R_1,R_s)&=&s M_m(s) \\
\kappa(R_1,R_2,R_s)&=&(s+2)\kappa(R_2,R_s)
\end{eqnarray}
and the initial conditions~(\ref{eq:luem1})--(\ref{eq:luem2}), to
\begin{eqnarray}
(s+2)\kappa(R_2,R_s)&=&2M_m(2)\left(M_{m+1}(s)+M_{m-1}(s)-2M_m(s)\right) \nonumber\\
&&+~\!sM_m(1)\left(M_{m+1}(s)-M_{m-1}(s)\right). \label{eq:LUEt1t2}  
\end{eqnarray}
The mixed difference equation~\eqref{eq:luemixed} is then established by substituting~\eqref{eq:lueloop1} into~\eqref{eq:LUEt1t2} and simplifying the resulting equation using the obtained difference equation~\eqref{eq:LUEMdiff}.

We now derive the difference equation \eqref{eq:LUE-s-recurrence} in~\cite{HT03} from the difference equations~\eqref{eq:LUEMdiff} and~\eqref{eq:luemixed}. Solving \(M_{m+1}(s)\) and \(M_{m-1}(s)\) from the equations~\eqref{eq:LUEMdiff} and~\eqref{eq:luemixed} gives
\begin{eqnarray}
M_{m+1}(s)&=&M_m(s)+\frac{(s+2)M_m(s+1)+s\left(s+1-2m-\alpha\right)M_m(s)}{2m(m+\alpha)} \label{eq:LUE-shift-plus}\\
M_{m-1}(s)&=&M_m(s)+\frac{-(s+2)M_m(s+1)+s\left(s+1+2m+\alpha\right)M_m(s)}{2m(m+\alpha)}. \label{eq:LUE-shift-minus}
\end{eqnarray}
Shifting \(s\to s+1\) in~\eqref{eq:LUE-shift-plus} and \(m\to m+1\) in~\eqref{eq:LUE-shift-minus} before applying~\eqref{eq:LUE-shift-plus}
to the latter in eliminating \(M_{m+1}(s)\) leads to
\begin{eqnarray*}
M_{m+1}(s+1)&=&\frac{s+3}{2m(m+\alpha)}M_m(s+2)+\left(1+\frac{(s+1)(s+2-2m-\alpha)}{2m(m+\alpha)}\right)M_m(s+1) \label{eq:LUE-shift-s} \\ 
M_{m+1}(s+1)&=&\left(1+\frac{(s+2)(2m+\alpha+s+1)}{2m(m+\alpha)}\right)M_m(s+1)+\frac{s\left((s+1)^2-\alpha^2\right)}{2m(m+\alpha)}M_m(s), \label{eq:LUE-shift-m}
\end{eqnarray*}
from which one immediately recovers the difference equation \eqref{eq:LUE-s-recurrence} in $s$ obtained in~\cite{HT03}. This completes the proof of Proposition~\ref{prop4}.
\end{proof}

\subsection{Mixed difference equation of Bures--Hall ensemble}\label{sec:3.2}
\begin{proposition}\label{prop5}
The average spectral moment~(\ref{eq:spM}) of the Bures--Hall ensemble~(\ref{eq:UBH}) satisfies the mixed difference\footnote{Expressions of the coefficients $f_i$ are omitted here for brevity, which are provided in the online repository at \href{https://doi.org/10.5281/zenodo.22133470}{doi.org/10.5281/zenodo.22133470}}  in $m$ and $s$,
\begin{equation}\label{eq:BHmixed}
f_1M_m(s+2)+f_2M_{m+1}(s+2)+f_3M_m(s)+f_4M_{m+1}(s)=0,
\end{equation}
combining which with the difference equation~(\ref{eq:M-recurrenceBH}) reproduces the known difference equation\footnote{Expressions of the coefficients $g_i$ are found in~\cite{WHW26}.} in $s$, 
\begin{equation}\label{eq:s-recurrenceBH}
g_1 M_m(s+2)=g_2M_m(s)+g_3M_m(s-2),
\end{equation}
\end{proposition}

\begin{proof}
The Virasoro constraints~\cite{MM90,AdlerVanMoerbeke1995} applied for the Bures--Hall case boil down to the integration by parts
\begin{equation}\label{eq:insx^3R_k}
0=\int_{\mathbb R_+^m}\sum_{i=1}^{m}
\frac{\partial}{\partial x_i}\left(x_i^3R_s~\!\rho(\mathbf{x})\right)\prod_{i=1}^m \dd x_i.
\end{equation}
Inserting the density~(\ref{eq:UBH}) into~(\ref{eq:insx^3R_k}) and evaluating the corresponding derivatives, similarly to (\ref{eq:3.1drho})--(\ref{eq:pair-sum}), one has
\begin{equation}\label{eq:insx^3R_ksimplify}
0=\int_{\mathbb R_+^m}
\left(sR_{s+2}+(m+\alpha)R_sR_2+\frac{3}{2}R_sR_1^2-R_sR_3\right)
\rho(\mathbf{x})\prod_{i=1}^m \dd x_i.
\end{equation}
Using the moment-cumulant relation~\cite{PT11,HW26} before applying  Lemma~\ref{lemma1} to simplify joint cumulants involving $R_1$, the above equation~(\ref{eq:insx^3R_ksimplify}) turns into
\begin{equation}\label{eq:virasoroid}
\kappa\!\left(R_s,R_3\right)=(m+\alpha)\kappa\!\left(R_s,R_2\right)+sM_m\!\left(s+2\right)+\frac{3s}{2}\left(m\left(m+2\alpha\right)+s+1\right)M_m(s),
\end{equation}
where have utilized the initial conditions $M_m(i)$, $i=1,2,3$, obtained from~\cite{HT03,HW26}. 

The covariance in~(\ref{eq:virasoroid}) can be eliminated by introducing higher flows of the B-Toda hierarchy~\cite{Date1982,Jimbo1983,AdlerVanMoerbeke1995,AdlerVanMoerbeke1997, AdlerHorozovVanMoerbeke1999,HuLi17,LiYu2022}, where we need the deformation involving $t_i$, $i=1,2,3$, as
\begin{equation}\label{eq:bh-extended-tau}
\tau_m\!\left(t_1,t_2, t_3,q;s\right)=\frac{Z}{m!}\int_{\R_+^m}\rho(\mathbf{x})\prod_{i=1}^m\e^{t_1x_i+t_2x_i^2+t_3x_i^3+q x_i^s}\dd x_i.
\end{equation}
By applying the de Bruijn's formulae~\cite{deBruijn1955} and Schur's Pfaffian identities~\cite{Schur1911}, the signed tau function
\begin{equation}\label{eq:tautilde1}
\Tilde{\tau}_m=(-1)^{\lfloor\frac{m}{2}\rfloor}\tau_m,
\end{equation}
admits a Pfaffian representation~\cite{HuLi17}
\begin{eqnarray}
\Tilde{\tau}_{2m} &=& \operatorname{Pf}(0,1,\cdots,2m-1) \label{eq:sec3Pfa} \\
\Tilde{\tau}_{2m+1} &=& \operatorname{Pf}\left(d_0,0,1,\cdots,2m\right) \label{eq:sec3Pfa2}
\end{eqnarray}
with entries 
\begin{eqnarray}
\!\!\!\!\!\!\!\!\!\!\!\!\!\operatorname{Pf}(i,j) &=& \iint_{\mathbb{R}_{+}^2} x^{i+\alpha-\frac{1}{2}} y^{j+\alpha-\frac{1}{2}}\frac{x-y}{x+y}\e^{-(x+y)+\sum_{r=1}^3 t_r(x^r+y^r)+q(x^s+y^s)} \dd x \dd y \\
\!\!\!\!\!\!\!\!\!\!\!\!\!\operatorname{Pf}\left(d_0, i\right)\!&=&\!\int_{\mathbb{R}_{+}} \!\!x^{i+\alpha-\frac{1}{2}} \e^{-(1-t_1)x+t_2 x^2+t_3 x^3+qx^s} \dd x.\label{eq:sec3Pfabeta}
\end{eqnarray}
These Pfaffian entries satisfy the dispersion relations
\begin{eqnarray}\label{eq:dispersion1}
\frac{\partial}{\partial t_k}\operatorname{Pf}(d_0,i)
&=& \operatorname{Pf}(d_0,i+k)\\
\frac{\partial}{\partial t_k}\operatorname{Pf}(i,j)
&=& \operatorname{Pf}(i+k,j)+\operatorname{Pf}(i,j+k),
\label{eq:dispersion2}
\end{eqnarray}
which is known as the standard skew-moment evolution underlying the Pfaff lattice~\cite{AdlerHorozovVanMoerbeke1999}. The Pfaffian representation~(\ref{eq:sec3Pfa})--(\ref{eq:sec3Pfa2}), together with the dispersion relations~(\ref{eq:dispersion1})--(\ref{eq:dispersion2}), place the signed tau function within the charged-BKP hierarchy~\cite{KacvandeLeur98,vandeluer01}. The Hirota bilinear equations of the charged-BKP hierarchy are encoded in the generating series~\cite{KacvandeLeur98} of dimensions $m_1$ and $m_2$ as
\begin{eqnarray}\label{eq:geneseries}
\!\!\!\!\!\!\!\! 0 &=& \frac{1}{2}\left((-1)^{m_1+m_2}\!-\!1\right) \e^{\sum_{r=1}^{\infty} y_r D_{t_r}} \Tilde{\tau}_{m_1} \cdot \Tilde{\tau}_{m_2}+\sum_{j=0}^{\infty}\left( P_j(2 \mathbf{y}) P_{j+m_1-m_2-1}(-D) \right.\nonumber \\
\!\!\!\!\!\!\!\!&& \!\!\!\!  \!\!\!\! \!\!\!\! \!\!\!\! \!\!\!\!\left.\times \e^{\sum_{r=1}^{\infty} y_r D_{t_r}} \Tilde{\tau}_{m_1-1} \cdot \Tilde{\tau}_{m_2+1}+ P_j(-2 \mathbf{y}) P_{j+m_2-m_1-1}(D) \e^{\sum_{r=1}^{\infty} y_r D_{t_r}} \Tilde{\tau}_{m_1+1} \cdot \Tilde{\tau}_{m_2-1}\right),
\end{eqnarray}
where $P_i$ denotes the elementary Schur functions~\cite{Kac90},
\begin{eqnarray}
\sum_{i\in \mathbb{Z}}P_i({\bf z})\lambda^i=\e^{\sum_{r \geq 1} z_r \lambda^r}
\end{eqnarray}
and
\begin{equation}
D=\left(D_{t_1},\frac{D_{t_2}}{2},\frac{D_{t_3}}{3},\ldots\right).
\end{equation}
Now specializing the generating series~(\ref{eq:geneseries}) to the case
$m_2=m_1+3$ and $\mathbf{y}=(y,0,0,\ldots)$,  expansion in $y$ in the resulting expression, together with the sign change~(\ref{eq:tautilde1}), yields
\begin{eqnarray}\label{eq:y^0y^1}
0&=&\left(-\tau_{m_1} \tau_{{m_1}+3}+P_2(D) \tau_{{m_1}+1}\cdot\tau_{{m_1}+2}\right) \nonumber \\
&&~\!+y\left(-D_{t_1}\tau_{m_1}\cdot \tau_{m_1+3}+\left(P_2(D)D_{t_1}-2P_3(D)\right)\tau_{m_1+1}\cdot\tau_{m_1+2}\right)+\mathcal{O}\!\left(y^2\right).
\end{eqnarray} 
Setting $m_1 =m-1$ in~(\ref{eq:y^0y^1}) and collecting the coefficients of $y^0$ and $y^1$, we obtain
\begin{equation}\label{eq:bkp1} 
\left(D_{t_1}^2+D_{t_2}\right)\tau_m\cdot\tau_{m+1}=2\tau_{m+2}\tau_{m-1},
\end{equation}
and
\begin{equation}\label{eq:bkp2} 
\left(4D_{t_3}+3D_{t_1}D_{t_2}-D_{t_1}^3\right)\tau_m\cdot\tau_{m+1} = 6D_{t_1}\tau_{m+2}\cdot\tau_{m-1},
\end{equation}
respectively. 

Taking derivatives of $q$ on both sides of (\ref{eq:bkp1}) and evaluating at $q=t_i=0$, $i=1,2,3$, gives
\begin{equation}\label{eq:1sthirota}
\kappa\!\left(R_s,R_2\right)-\kappa_{(m+1)}\!\left(R_s,R_2\right)=\sum_{i=-1}^{2}a_{i}M_{m+i}(s),
\end{equation}
where the same procedure applied to~(\ref{eq:bkp2}) gives 
\begin{equation}\label{eq:2ndhirota}
4\left(\kappa\!\left(R_s,R_3\right)-\kappa_{(m+1)}\!\left(R_s,R_3\right)\right)+c_0\kappa\!\left(R_s,R_2\right)+c_1\kappa_{(m+1)}\!\left(R_s,R_2\right)
=\sum_{i=-1}^{2} b_{i}M_{m+i}(s)
\end{equation}
with the coefficients $c_0$, $c_1$, $a_i$, $b_i$ omitted being rational functions of $m$ and $s$. Solving~(\ref{eq:virasoroid}) and~(\ref{eq:2ndhirota}), we can represent the difference of 
\begin{equation}
\kappa\!\left(R_s,R_2\right)
\end{equation}
and
\begin{equation}
\kappa_{(m+1)}\!\left(R_s,R_2\right)
\end{equation}
in terms of a linear combination of $M_{m}\!\left(s+2\right)$, $M_{m+1}\left(s+2\right)$, and $M_{m+i}(s)$, $i=-1,0,1,2$, independent from the difference of the two covariances~(\ref{eq:1sthirota}). Solving this two-by-two equation system gives
\begin{equation}\label{eq:rkr2mkk+2}
\kappa(R_s,R_2)=u_{0}M_m(s+2)+ u_{1}M_{m+1}(s+2)+\sum_{i=-1}^{2}v_{i}M_{m+i}(s),
\end{equation}
where $u_0$, $u_1$ and $v_i$ are certain rational coefficients. On the other hand, similar to~(\ref{eq:LUEkappads}), we have from definition~(\ref{eq:bh-extended-tau}) that
\begin{equation}
\left.\frac{\partial^4}{\partial t_1^2 \partial t_2 \partial q}\ln\tau_m\!\left(t_1,t_2,t_3,q;s\right)\right\rvert_{t_1=t_2=t_3=q=0}=(s+2)(s+3)\kappa(R_s,R_2).
\end{equation}
Since the tau function~(\ref{eq:bh-extended-tau}) also satisfies the B-Toda equation~(\ref{eq:log-B-Toda}), taking derivatives of~(\ref{eq:log-B-Toda}) with respect to $q$, $t_1$, and $t_2$, with relevant terms evaluated using Lemma~\ref{lemma1}, one has
\begin{eqnarray}\label{eq:btodaflow}
&&\frac{m(m+2\alpha)}{4(m+\alpha)}\left(\left(2m+2\alpha+s+2\right)\kappa_{(m+1)}\!\left(R_s,R_2\right)+\left(2m+2\alpha-s-2\right)\kappa_{(m-1)}(R_s,R_2)\right)\nonumber\\&=&
\left(m(m+2\alpha)+\left(s+2\right)\left(s+3\right)\right)
\kappa\!\left(R_s,R_2\right)-\frac{12s\left(s+1\right)}{m(m+2\alpha)}M_m(s)M_m(2)\nonumber\\&&+
\frac{s\,m(m+2\alpha)}{4(m+\alpha)^2}
\left(M_{m+1}(s)-M_{m-1}(s)\right)
\left(M_{m+1}\!\left(2\right)-M_{m-1}\!\left(2\right)\right).
\end{eqnarray}
Inserting $m$ shifted versions of~(\ref{eq:rkr2mkk+2}) into~(\ref{eq:btodaflow}), we obtain
\begin{equation}\label{eq:bhmixedorigin}
\sum_{i=-1}^{2}A_{i}M_{m+i}\!\left(s+2\right)
=\sum_{i=-2}^{3}B_{i}M_{m+i}\!\left(s\right),
\end{equation}
where $A_i$ and $B_i$ denote coefficients omitted here for brevity. Repeated use of the difference equation~(\ref{eq:M-recurrenceBH}) of shifted $m$ to simplify~(\ref{eq:bhmixedorigin}), the mixed difference equation~(\ref{eq:BHmixed}) is established. 

We now derive the difference equation \eqref{eq:s-recurrenceBH} in~\cite{WHW26} from the difference equations~(\ref{eq:BHmixed}) and~(\ref{eq:M-recurrenceBH}). This is done similarly to the proof of Proposition \ref{prop4} by inserting~(\ref{eq:M-recurrenceBH}) into appropriately $s$-shifted versions of~(\ref{eq:BHmixed}) to eliminate all but terms of dimension $m$. This completes the proof of Proposition~\ref{prop5}.
\end{proof}

\section{Discussion}\label{sec:con}
Beyond average entropies studied in this work, a natural question as posed in the Introduction is if integrable systems methods are able to address the corresponding higher-order cumulants. The discussion here mainly focuses on von Neumann entropy over Hilbert--Schmidt ensemble~(\ref{eq:LUE}), where expressions of the first six cumulants have been obtained via random matrix methods~\cite{Page93,Foong94,Ruiz95,Se96,VPO16,Wei17,Wei20,HWC21,HW26}. Once this setting is understood, results of other ensembles considered may follow in parallel.

If one seeks a direct extension of the tau function~(\ref{eq:tau-introduction}) for an arbitrary $l$-th order cumulant, one may consider the tau function
\begin{equation}\label{eq:tau-discussion}
\tau_m({\bf t},{\bf q};{\bf s})=\frac{Z}{m!}\int_{\mathbb{R}_+^m}\rho(\mathbf{x})\prod_{i=1}^m\e^{\sum_{j=1}^{l}q_j x_i^{s_j}+\sum_{r\geq 1}t_rx_i^r}\dd x_i
\end{equation}
of real-valued power parameters ${\bf s}=(s_1,\dots,s_l)$ that is responsible for generating the joint cumulant \(\kappa(R_{s_1},\ldots,R_{s_l})\). The key observation is that differentiation with respect to power parameters preserves Toda lattice structures inherited from the underlying determinantal identities. This suggests a recursive scheme in the cumulant order $l$ with each step reduced to solving a second-order difference equation in the matrix dimension $m$. However, the recursive scheme does not fully capture the simplicity of higher-order cumulant expressions. In particular, it does not appear sufficient to prove the following conjecture on the difference equation in matrix dimension satisfied by cumulant of any order, inferred from the first six cumulants summarized in~\cite{HW26}.
\begin{conjecture}
For any $l>2$, the $l$-th cumulant $\kappa_{l}(T)$ of the induced entropy~(\ref{eq:T}) satisfies the $l$-th order difference equation
\begin{equation}\label{eq:entropy-cumulant-conjecture}
\Delta^{l}\kappa_{l}(T)=l!\left((l-1)\psi_{l-2}(n)+n\psi_{l-1}(n)\right),
\end{equation}
where we denote $\Delta f_m=f_{m+1}-f_m$.
\end{conjecture}

We notice that the right-hand side of~(\ref{eq:entropy-cumulant-conjecture}) involving polygamma functions~(\ref{eq:polygamma}) is a universal constant independent of $m$, implying that $\kappa_{l}(T)$ is a polynomial in $m$ of degree $l$. Resolving this conjecture requires a more intricate tau function than the simple construction~(\ref{eq:tau-discussion}). Specifically, a promising tau function shall be able to circumvent the need to compute an increasing number of initial conditions as $l$ increases while capturing cumulant structure of the conjecture. This requires that the tau function uncovers relations among cumulants of different orders $l$ in addition to lattice equations in matrix dimension $m$. The underlying deformed orthogonal polynomials of tau functions to be constructed may also play a role. Beyond the limitations of random matrix methods~\cite{HW26}, the ultimate goal of integrable systems approach is to read off the cumulant expression of an arbitrary order from tau functions without resorting to any recursive scheme.

Progress in higher-order cumulants of von Neumann entropy can also be applied to the study of higher-order correlators of spectral moments, extending results in the literature~\cite{CundenMezzadriOConnellSimm,DY17,GGR20,GGR21} from integer-valued to real-valued power parameters. The focus is on establishing recurrence relations among correlators of different power parameters and matrix dimensions instead of their summation representations. Such relations, once established, will provide a systematic framework of finding higher-order cumulants of more general and challenging entanglement metrics including Tsallis entropy, entanglement capacity, and relative entropy.

\begin{appendices}
\section{Trace-scaling identity of joint cumulants}\label{app1}
Here, we summarize in Lemma~\ref{lemma1} below a new cumulant identity useful in simplifying joint cumulants resulting from Toda lattice equations and Virasoro constraints in the main text. Applicable densities of Lemma~\ref{lemma1} include~(\ref{eq:LUE}), (\ref{eq:UBH}), (\ref{eq:LMB}), and (\ref{eq:uBKM}).

\begin{lemma}\label{lemma1}
For a joint density of random variables $\mathbf{x}=\{x_1,\dots,x_m\}$ of the form
\begin{equation}
\rho(\mathbf{x})=
h_m(\mathbf{x})\e^{-R_1}
\end{equation}
with \(h_m\) being a homogeneous function of degree \(d\),
\begin{equation}\label{eq:density-homogeneity}
h_m(c\mathbf{x})=c^{d}h_m(\mathbf{x}),
\end{equation}
a joint cumulant~(\ref{eq:joint-cumulant-def}) of $n$ traces $R_1$ and $l$ spectral moments~(\ref{eq:Rs}) of real-valued parameters \(s_1,\ldots,s_l\) is scaled as
\begin{equation}\label{eq:general-cumulant-scaling}
\kappa\!\left(\underbrace{R_1,\ldots,R_1}_{n},
R_{s_1},\ldots,R_{s_l}\right)=\left(\sum_{i=1}^{l}s_i\right)_{\!\!n}\kappa\left(R_{s_1},\ldots,R_{s_l}\right),
\end{equation}
where $(a)_n=\Gamma(a+n)/\Gamma(a)$ is the Pochhammer's symbol.
\end{lemma}

\begin{proof}
The tau function that generates the considered joint cumulant is
\begin{equation}\label{eq:general-deformed-partition}
\tau_m(t,\mathbf q;\mathbf s)=\frac{1}{m!}\int_{\R_+^m}h_m(\mathbf{x})\e^{-(1-t)R_1+\sum_{i=1}^{l}q_iR_{s_i}}\prod_{i=1}^{m}\dd x_i.
\end{equation}
The change of variables
\begin{equation}
y_i=(1-t)x_i
\end{equation}
and the homogeneity condition~\eqref{eq:density-homogeneity} give
\begin{equation}\label{eq:general-tau-scaling}
\tau_m(t,\mathbf q;\mathbf s)=(1-t)^{-\gamma}\tau_m\left(0,\frac{q_1}{(1-t)^{s_1}},\dots,\frac{q_l}{(1-t)^{s_l}};\mathbf s\right),
\end{equation}
where
\begin{equation}
\gamma=d+m.
\end{equation}
Differentiating $\ln\tau_m$ with respect to each \(q_j\) before setting \(\mathbf q=\mathbf 0\), we obtain
\begin{equation}\label{eq:cumulant-scaling-intermediate}
\left.\frac{\partial^l}{\partial q_1\cdots \partial q_l}\ln\tau_m(t,\mathbf q;\mathbf s)\right\rvert_{\mathbf q=\mathbf 0}
=(1-t)^{-\sum_{i=1}^{l}s_i}\kappa\left(R_{s_1},\ldots,R_{s_l}\right).
\end{equation}
Finally, applying $\frac{\partial^n}{\partial t^n}$ to both sides of \eqref{eq:cumulant-scaling-intermediate} and then setting \(t=0\), the left-hand side becomes
\begin{equation}
\kappa\left(\underbrace{R_1,\ldots,R_1}_{n},R_{s_1},\ldots,R_{s_l}\right),
\end{equation}
whereas the fact that 
\begin{equation}
\left.\frac{\partial^n}{\partial t^n}(1-t)^{-\sum_{i=1}^{l}s_i}\right\rvert_{t=0}=\left(\sum_{i=1}^{l}s_i\right)_{\!\!n} 
\end{equation}
establishes \eqref{eq:general-cumulant-scaling}. This completes the proof of Lemma~\ref{lemma1}.
\end{proof}
\end{appendices}

\noindent\textbf{Acknowledgment}~~Lu Wei was supported by the U.S. National Science Foundation (2306968) and the U.S. Department of Energy (DE-SC0024631). Peter J. Forrester was supported by the Australian Research Council (DP250102552).\\

\end{document}